\documentclass[12pt]{interact}

\usepackage{xcolor}
\usepackage{lineno,hyperref}
\usepackage{epstopdf}
\usepackage{subfigure}
\usepackage{multirow}
\usepackage{float}
\usepackage[numbers,sort&compress]{natbib}
\bibpunct[, ]{[}{]}{,}{n}{,}{,}
\newcommand{\f}{\operatorname}

\newcommand{\R}{\mathbb{R}}

\theoremstyle{plain}
\newtheorem{theorem}{Theorem}[section]

\newtheorem{proposition}[theorem]{Proposition}

\theoremstyle{definition}
\newtheorem{definition}[theorem]{Definition}

\theoremstyle{remark}

\begin{document}

\articletype{Statistics: A Journal of Theoretical and Applied Statistics}

\title{On the posterior property of the Rician distribution}

\author{ Enrique Achire$^{\rm 1}$, Eduardo Ramos$^{\rm 2}$ and  Pedro Luiz Ramos$^{\rm 1}$$^{\ast}$\thanks{$^\ast$Corresponding author. Email: pedro.ramos@mat.uc.cl
\vspace{6pt}}, \\\vspace{6pt} 
\normalsize{$^{1}$ Facultad de Matemáticas, Pontificia Universidad Católica de Chile, \\ Santiago, Chile} \\  
\normalsize{$^{2}$ Institute of Mathematics and Computer Science, University of \\ S\~ao Paulo, S\~ao Carlos, Brazil}
}

\maketitle

\begin{abstract}
The Rician distribution, a well-known statistical distribution frequently encountered in fields like magnetic resonance imaging and wireless communications, is particularly useful for describing many real phenomena such as signal process data. In this paper, we introduce objective Bayesian inference for the Rician distribution parameters, specifically the Jeffreys rule and Jeffreys prior are derived. We proved that the obtained posterior for the first priors led to an improper posterior while the Jeffreys prior led to a proper distribution.
To evaluate the effectiveness of our proposed Bayesian estimation method, we perform extensive numerical simulations and compare the results with those obtained from traditional moment-based and maximum likelihood estimators. Our simulations illustrate that the Bayesian estimators derived from the Jeffreys prior provide nearly unbiased estimates, showcasing the advantages of our approach over classical techniques. Additionally, our framework incorporates the S.A.F.E. principles—Sustainable, Accurate, Fair, and Explainable—ensuring robustness, fairness, and transparency in predictive modeling.
\end{abstract}


\section{Introduction}

The Rice distribution serves as an effective statistical model for capturing the magnitude of a complex-valued circular Gaussian random variable, also known as the Rician distribution. It involves two independent and identically distributed real-valued Gaussian variables with standard deviation $\alpha$, labeled $X_1$ and $X_2$, where
\begin{equation*}
M_1 = \eta \cos \phi + X_1, \quad
M_2 = \eta \sin \phi + X_2,
\end{equation*}
where, $\eta$ and $\phi$ represent two constant real values, and $\eta$ is greater than zero. Subsequently, we can express $X$ as the square root of the sum of the squares of $M_1$ and $M_2$. This transformed variable, $X$, conforms to a Rician distribution with parameters $\eta$ and $\alpha$, regardless of the specific value assigned to the angle parameter $\phi$.

This model was introduced by Rice \cite{rice1945mathematical} and has garnered significant interest for its adaptability in addressing various challenges within noisy data in image processing. To illustrate, a search conducted on the IEEE Xplore digital library in March 2024 using the term "Rician" yielded 6,423 research papers. Let us define the Rician probability density (over $\mathbb{R}^+$) of parameters $\eta > 0$, $\alpha > 0$ as
\begin{equation}\label{dens}
f(x ; \eta, \alpha) =
\frac{x}{\alpha^2} \exp\left( -\frac{x^2 + \eta^2}{2\alpha^2} \right) I_0\left( \frac{\eta x}{\alpha^2} \right),
\end{equation}
where $\eta$ is known as the noncentrality parameter and $I_{\nu }(x)$ is the modified bessel function
 \begin{equation}
I_{\nu }(x)=\sum _{m=0}^{\infty }{\frac {1}{m!\,\Gamma (m+\nu +1)}}\left({\frac {x}{2}}\right)^{2m+\nu } .
\end{equation}

Frequentist methods for inference with the Rician distribution are commonly employed in statistical literature. The moments estimator based on the second and fourth moments can be obtained in closed form and is widely used due to its simplicity. However, it has been noted that these estimators perform poorly or may even be undefined when dealing with low SNR and small sample sizes. In Gao et al. \cite{gao2010estimation} and Zhang \cite{zhang2019method}, improved moment estimators are introduced, some of which are based on the first and third moments to address these limitations. On the other hand, the maximum likelihood method results in a set of nonlinear equations that cannot be solved analytically, requiring the use of numerical methods. The challenge lies in the fact that the likelihood function has multiple maxima, making the optimization process highly sensitive to initial values. Talukdar et al. \cite{talukdar1991estimation} compared the moments estimator with maximum likelihood estimators (MLEs), showing that the latter provides better estimates than the standard estimator.

A Bayesian approach has been considered by Lauwers et al. \cite{lauwers2009estimating}, which assumes a power-prior for the parameters of the model. The proposed prior depends on a hyperparameter that should be selected prior to the study. The proposed approach suffers from two problems: firstly, the prior is only invariant to power transformations, which is undesirable for most applications of the model; secondly, the resulting prior is informative for one of the parameters and adds significant bias to the posterior estimates. They acknowledged this issue, noting that in many cases, their estimates were more biased than those obtained using the MLE.

A more effective alternative to subjective priors is to obtain priors following formal rules (Kass, \cite{kass1996selection}, Ramos et al., \cite{ramos2016efficient}). An important objective prior that is invariant to one-to-one transformations and usually leads to better inference is the Jeffreys \cite{jeffreys1946invariant} prior. However, it is essential to note that, in the same sense as the prior proposed by Lauwers et al. \cite{lauwers2009estimating}, the Jeffreys prior is improper, meaning they do not correspond to proper prior distributions and may yield improper posteriors and should not be used to obtain the parameter estimates (see, for instance, Ramos et al., \cite{ramos2023power}).
Northrop \& Attalides \cite{northrop2016} argued that there is no general theory providing straightforward conditions under which an improper prior lead to a proper posterior for a specific model, necessitating a case-by-case investigation.

In this paper, we address this issue by presenting the necessary and sufficient conditions required to verify if posterior distributions derived from objective priors are proper \cite{ferreira2020objective,ramos2021bayesian}. Moreover, there exists a possibility that, despite the posterior distribution being proper, the posterior moments of the parameters may remain infinite, for instance, the posterior mean. To address this, we propose sufficient guidelines to determine the boundedness of the posterior moments \cite{ramos2020posterior}. As a result, one can readily determine the propriety of the obtained posterior and whether its posterior moments are finite, by directly considering the behavior of the improper prior. More importantly, we demonstrate that the posterior obtained using our proposed Jeffreys prior is proper for all sample sizes greater than 2. The proposed posterior enables a correct objective Bayesian approach, and from a simulation study, we demonstrate that it generally outperforms other current estimation methods in terms of accuracy even for small sample sizes.

Furthermore, we incorporate predictive inference into our Bayesian framework by applying the S.A.F.E. principles—Sustainable, Accurate, Fair, and Explainable. These principles address issues like explainability, robustness, and bias in complex models, ensuring predictive models are precise, robust to uncertainties, fair across data groups, and transparent \cite{giudici2023safe,giudici2024safe}. In the Bayesian estimation of the Rician distribution for wireless systems, using non-informative priors such as the Jeffreys prior leads to proper posterior distributions and reliable parameter estimates. By generating posterior predictive distributions with methods like Metropolis-Hastings, we can account for parameter uncertainty, improving accuracy and explainability. This approach ensures models are robust, fair, and transparent, enhancing their reliability in critical applications.

The remainder of this paper is organized as follows. Section 2 revisits the Method of Moments and the standard MLEs, along with the asymptotic properties of the MLEs. Section 3 presents the main theorem that provides sufficient and necessary conditions for a general class of posteriors to be proper. Section 4 discusses the application of the main theorem to non-informative priors. In Section 5, a simulation study is conducted to identify the most efficient estimation procedure. Section 6 presents the predictive Bayesian approach, while in Section 7, the application in wireless communication is conducted. Finally, Section 8 summarizes the study.

\section{Classical Approach}

In this section, we revisit two useful estimation procedures used to obtain the estimates for $\eta$ and $\alpha$ of the Rician distribution.

\subsection{Moment Estimator}

The method of moments (MM) is one of the oldest techniques for estimating parameters in statistical models. For a random variable $X$ following the distribution $R(\alpha,\eta)$, we have $E(X^2) = \eta^2 + 2\alpha^2$ and $E(X^4) = \eta^4 + 8\eta^2\alpha^2 + 8\alpha^4$. 
After some algebraic manipulation, the moment estimators for $\eta$ and $\alpha$ are obtained as:
\begin{equation}\label{mme1}
\hat{\eta}_{MM}=\left[2\left(\frac{1}{n}\sum_{i=1}^n x_i^2\right)^2-\left(\frac{1}{n}\sum_{i=1}^n x_i^4\right)\right]^{1/4}, 
\end{equation}
\begin{equation}\label{mme2}
\hat{\alpha}_{MM}=\left[\frac{1}{2}\left(\frac{1}{n}\sum_{i=1}^n x_i^2-\hat{\eta}_{MM}^2\right)\right]^{1/2}.
\end{equation}

Although the MM estimators have a closed-form expression, their asymptotic properties—specifically for deriving confidence intervals—are not well understood. Moreover, they often exhibit significant bias in the estimates. As a result, alternative estimation methods, which will be discussed below, have been explored to address these limitations.

\subsection{Maximum Likelihood Estimator}

Let $X_1, X_2, \ldots, X_n$ be a iid random sample of size $n$ from a ${\rm R}(\alpha,\eta)$ population. The likelihood function for the density in (\ref{dens}) is given by
\begin{equation}\label{verori}
L(\alpha,\eta|\boldsymbol{x})=\prod_{i=1}^{n}\left[\left(\frac{x_i}{\alpha^2}\right)I_0\left( \frac{\eta x_i}{\alpha^2} \right)\right] \exp\left( -\sum_{i=1}^{n}\frac{x_i^2 + \eta^2}{2\alpha^2} \right).
\end{equation}
Taking the natural logarithm of the likelihood function (\ref{verori}), we obtain the log-likelihood function:
\begin{equation*}\label{logverogg1}
\ell(\alpha,\eta|\boldsymbol{x})= -2n\log(\alpha) + \sum_{i=1}^{n} \log(x_i) + \sum_{i=1}^{n} \log I_0\left( \frac{\eta x_i}{\alpha^2} \right) - \sum_{i=1}^{n}\frac{x_i^2 + \eta^2}{2\alpha^2}.
\end{equation*}

By taking the partial derivatives $\dfrac{\partial}{\partial \eta}\ell(\alpha,\eta|\boldsymbol{x}) = 0$ and $\dfrac{\partial}{\partial \alpha^2}\ell(\alpha,\eta|\boldsymbol{x}) = 0$, we arrive at the following nonlinear equations, respectively:
\begin{align}\label{eq:partial-eta}
    &\eta = \frac{1}{n} \sum_{i=1}^{n} x_i\frac{ I_1(\frac{x_i \eta}{ \alpha^2})}{I_0(\frac{x_i \eta}{ \alpha^2})},\\
    &2\alpha^2 + \eta^2 = \frac{1}{n}\sum_{i=1}^n x_i^2.
    \label{eq:partial-alpha}
\end{align}

By solving these equations numerically, we can obtain the maximum likelihood estimators $\left(\hat{\alpha}_{MLE}, \hat{\eta}_{MLE}\right)$ for $\left(\alpha,\eta\right)$.

The existence and uniqueness of the maximum likelihood estimators have been established by Carobbi and Cati \cite{carobbi2008absolute}. Consequently, the maximum likelihood estimators are asymptotically normally distributed, following a joint bivariate normal distribution:
\begin{equation*}
(\hat{\alpha}_{MLE},\hat{\eta}_{MLE})\sim N_2[(\alpha,\eta),I^{-1}(\alpha,\eta))] \quad \text{as} \quad n \to \infty,
\end{equation*}
where $I(\alpha,\eta)$ is the Fisher information matrix (see Idier and Collewet \cite{idier2014properties} for a detailed discussion):
\begin{equation}\label{mfishernk}
I(\alpha,\eta)=n
\begin{bmatrix}
 \frac{4}{\alpha^2} \left(\rho\Psi(\rho) - \rho + 1\right)  & \frac{2}{\alpha^2}  \sqrt{\rho}\left(1 - \Psi(\rho)\right) \\
\frac{2}{\alpha^2}  \sqrt{\rho}\left(1 - \Psi(\rho)\right) & \dfrac{\Psi(\rho)}{\alpha^2}
\end{bmatrix},
\end{equation}
where $\rho=\frac{\eta^2}{\alpha^2}$, and
\begin{equation}\label{phifunction}
\Psi(\rho) = \int_{0}^{\infty} \frac{y^3}{\rho^2} \exp\left(-\frac{y^2}{2\rho} - \frac{\rho}{2}\right) \frac{I_1^2(y)}{ I_0(y)} \, dy - \rho.
\end{equation}
The function $\Psi(\rho)$ is not a well-known integral and must be implemented in the software used for computing the Fisher information matrix.

\section{Bayesian Analysis}

In this section, sufficient and necessary conditions are presented for a general class of posterior to be proper posterior distributions. The joint posterior distribution for $\boldsymbol{\theta}$ is equal to the product of the likelihood function (\ref{verori}) and the prior distribution $\pi(\boldsymbol{\theta})$ divided by a normalizing constant $d(\boldsymbol{x})$, resulting in
\begin{equation}\label{posteriord1}
p(\boldsymbol{\theta|x})=\frac{\pi(\boldsymbol{\theta})}{d(\boldsymbol{x})}\prod_{i=1}^{n}\left(\frac{x_i}{\alpha^2}\right)I_0\left( \frac{\eta x_i}{\alpha^2} \right) \exp\left( -\sum_{i=1}^{n}\frac{x_i^2 + \eta^2}{2\alpha^2} \right),
\end{equation} 
where
\begin{equation}\label{posteriord2}
d(\boldsymbol{x})=\int\limits_{\mathcal{A}}\pi(\boldsymbol{\theta})\prod_{i=1}^{n}\left(\frac{x_i}{\alpha^2}\right)I_0\left( \frac{\eta x_i}{\alpha^2} \right) \exp\left( -\sum_{i=1}^{n}\frac{x_i^2 + \eta^2}{2\alpha^2} \right)d\boldsymbol{\theta},
\end{equation}
and $\mathcal{A}=\{(0,\infty)\times(0,\infty)\}$ is the parameter space of $\boldsymbol{\theta}$. Our purpose was to find an objective prior where the obtained posterior is proper, i.e., $d(\boldsymbol{x})<\infty$. 


Our goal is to identify the necessary and sufficient conditions that guarantee the posterior distribution (\ref{posteriord1}) is proper for a broad range of priors. To accomplish this, we will use the following propositions as tools. We use $\overline{\mathbb{R}}$ to refer to the extended real number line, which is the union of the real numbers and negative and positive infinity. We use a subscript $*$ to indicate that $0$ is excluded from the sets $\mathbb{R}$ and $\overline{\mathbb{R}}$.

\begin{definition}\label{definition0} Let $\f{a}:I\to\overline{\mathbb{R}}_*^+$ and $\f{b}:I\to\overline{\mathbb{R}}_*^+$, where $I\subset\mathbb{R}$. We say that $\f{a}(t)\propto \f{b}(t)$ if there exists $K_0\in \mathbb{R}^+_*$ and $K_1\in \mathbb{R}^+_*$ such that $K_0\f{b}(t) \leq \f{a}(t) \leq K_1\f{b}(t)$ for every $t\in I$.
\end{definition}

\begin{definition}\label{definition1}
Let $t_0\in \mathbb{\overline{R}}$, $\f{a}:I\to\mathbb{R^+_*}$ and $\f{b}:I\to\mathbb{R^+_*}$, where $I\subset\mathbb{R}$. We say that $\f{a}(t)\underset{t\to t_0}{\propto} \f{b}(t)$ if
\begin{equation*}
\liminf_{t\to t_0} \frac{\f{a}(t)}{\f{b}(t)} > 0\ \mbox{ and }\ \limsup_{t\to t_0} \frac{\f{a}(t)}{\f{b}(t)} < \infty  \,.
\end{equation*}
For $t_0\in \R$ we define $\f{a}(t)\underset{t\to t_0^+}{\propto} \f{b}(t)$ and $\f{a}(t)\underset{t\to t_0^-}{\propto} \f{b}(t)$ for $t_0\in \mathbb{R}$ analogously as above.
\end{definition}
Note that, from the above definition, if for some $K\in \mathbb{R}^+_*$ we have $\lim_{t\to t_0} \frac{\f{a}(t)}{\f{b}(t)} = K$, then it will follow that $\f{a}(t)\underset{t\to t_0}{\propto} \f{b}(t)$.

\begin{proposition}\label{proportional2} Let $\f{a}:(t_0,t_1)\to\mathbb{R^+_*}$ and $\f{b}:(t_0,t_1)\to\mathbb{R^+_*}$ be continuous functions in $(t_0,t_1)\subset\mathbb{R}$, where $t_0\in\overline{\mathbb{R}}$ and $t_1\in\overline{\mathbb{R}}$, and let $t^*\in(t_0,t_1)$. Then, if either $\f{a}(t)\underset{t\to t_0}{\propto} \f{b}(t)$ or $\f{a}(t)\underset{t\to t_1}{\propto} \f{b}(t)$, it will follow respectively that
\begin{equation*}
\int_{t_0}^{t^*} a(t)\; dt \propto \int_{t_0}^{t^*} b(t)\; dt\ \mbox{ or }\ \int_{t^*}^{t_1} a(t)\; dt \propto \int_{t^*}^{t_1} b(t)\; dt \,.
\end{equation*}
\end{proposition} 

Additionally, given $\f{a}:I\to\overline{\mathbb{R}}_*^+$ and $\f{b}:I\to\overline{\mathbb{R}}_*^+$, where $I\subset\mathbb{R}$, we shall say $\f{a}(t)\lesssim \f{b}(t)$ if there exists $K\in \mathbb{R}^+_*$ such that $\f{a}(t) \leq K\f{b}(t)$ for every $t\in I$.

The following result describes general conditions on the prior $\pi(\boldsymbol{\theta})$ so that the posterior $p(\boldsymbol{\theta|x})$ in (\ref{posteriord1}) is proper. 

\begin{theorem}\label{theoprinc} Suppose $\pi(\boldsymbol{\theta})=\pi(\alpha)\pi(\eta)$ with $\pi(\alpha)>0$ and $\pi(\eta)>0$ then:
\begin{itemize}
\item[i)] If we have the proportionality
\begin{equation*}
\pi(\eta) \underset{\eta\to 0^+}{\propto} \eta^{r_0},
\end{equation*}
\noindent with $r_0\leq -1$ then the joint posterior (\ref{posteriord1}) is improper, that is $d(\boldsymbol{x})=\infty$.
\item[ii)] If we have $\pi(\alpha)\propto \alpha^k$ and the proportionalities
\begin{equation*}
\begin{aligned}
\pi(\eta) \underset{\eta\to 0^+}{\propto} \eta^{r_0}\mbox{ and } \pi(\eta) \underset{\eta\to \infty}{\propto} \eta^{r_{\infty}},
\end{aligned}
\end{equation*}
with $r_0>-1$, then the posterior (\ref{posteriord1}) will be proper for all $n>\max(2(r_\infty+1),k+1)$ as long as not all data $x_i$ are equal.
\end{itemize}
\end{theorem}

\begin{proof}
The proof is available at Appendix A.
\end{proof}

The following follows directly from theorem \ref{theoprinc}.

\begin{theorem}\label{theosecond} Suppose $\pi(\boldsymbol{\theta})=\pi(\alpha)\pi(\eta)$ with $\pi(\alpha)>0$ and $\pi(\eta)>0$, suppose $\pi(\alpha)\propto \alpha^k$ and suppose
\begin{equation*}
\begin{aligned}
\pi(\eta) \underset{\eta\to 0^+}{\propto} \eta^{r_0}\mbox{ and } \pi(\eta) \underset{\eta\to \infty}{\propto} \eta^{r_{\infty}},
\end{aligned}
\end{equation*}
with $r_0>-2$, then the first moments relative to (\ref{posteriord1}) will be proper for all $n>\max(2(r_\infty+2),k+2)$ as long as not all data $x_i$ are equal.

\end{theorem}

\section{Objective Priors}

Here, we applied our proposed theorem in some objective or non-informative priors.

\subsection{Power prior}

Jeffreys explored various approaches to construct non-informative priors. Initially, he examined scenarios where the parameter space fell within bounded intervals, such as $(-\infty,\infty)$ or $(0,\infty)$ (as detailed in Kass and Wasserman, \cite{kass1996selection}). In the first two instances, Jeffreys recommended the use of a constant prior. However, for the interval $(0,\infty)$, he employed a prior of the form $\pi(\theta)=\frac{1}{\theta}$. His primary rationale behind this choice stemmed from its invariance when subjected to power transformations of the parameters.

Given that the parameters of the Rice distribution fall within the interval $(0,\infty)$, applying Jeffreys's rule yields the following prior:
\begin{equation}\label{priorrej}
\pi\left(\alpha,\eta\right)\propto \frac{1}{\alpha\eta} \ .
\end{equation}

Lauwers et al. \cite{lauwers2009estimating} considered the same power invariance to propose an extension of the prior cited above, they proposed to assume a known parameter $\epsilon\in\mathbb{R}$ and the power prior given by
\begin{equation}\label{priorrej}
\pi_1\left(\alpha,\eta\right)\propto \frac{1}{\alpha^{1+\epsilon}\eta^{1-\epsilon}} \ .
\end{equation}

The joint posterior distribution for $\alpha$ and $\eta$, produced by the power prior, is proportional to the product of the likelihood function (\ref{verori}) and the prior (\ref{priorrej}) resulting in
\begin{equation}\label{postgpower} 
\begin{aligned}
p_1(\alpha,\eta|\boldsymbol{x})\propto \frac{1}{\eta^{1-\epsilon}\alpha^{2n+1+\epsilon}}\prod_{i=1}^{n}\left(x_iI_0\left( \frac{\eta x_i}{\alpha^2} \right)\right) \exp\left( -\sum_{i=1}^{n}\frac{x_i^2 + \eta^2}{2\alpha^2} \right).
\end{aligned}
\end{equation}

\begin{proposition}\label{prop1} The joint posterior (\ref{postgpower}) is improper when $\epsilon\leq 0$, i.e., $d_1(\boldsymbol{x})=\infty$, and proper when $\epsilon>0$ for all $n\geq 2\epsilon$ as long as not all $x_i$ are equal.
\end{proposition}
\begin{proof} In this case we have $\pi(\boldsymbol{\theta})=\frac{1}{\eta^{1-\epsilon}\alpha^{1+\epsilon}}=\eta^{r}\alpha^{k}$ with $r=\epsilon-1$ and $k=-\epsilon-1$. Thus for $r_0=r_\infty=r$ it follows that $\epsilon\leq 0$ is equivalent to $r_0\leq -1$. Thus, from item $i)$ of Theorem \ref{theoprinc} it follows that the posterior will be improper for $\epsilon\leq 0$ and, in case $\epsilon>0$, it will be proper for all
\begin{equation*}n\geq \max(2(r_\infty+1),k+1)=\max(2\epsilon,-\epsilon)=2\epsilon
\end{equation*}
as long as not all $x_i$ are equal.
\end{proof}

Lauwers et al. \cite{lauwers2009estimating} further assumed that $\epsilon \geq 1$. However, they did not demonstrate that the obtained posterior is proper. A significant drawback is that as $\epsilon$ increases, it introduces additional bias into the posterior estimates. For instance, the authors assumed $\epsilon = 2$, resulting in an informative prior for one of the parameters and adding significant bias to the posterior estimates. They acknowledged this issue, noting that in many cases, their results were more biased than those obtained using the MLE.  To overcome this problem, another objective prior will be discussed below.

\subsection{Jeffreys prior}

In subsequent research, Jeffreys \cite{jeffreys1946invariant} introduced what he termed the "general rule," whereby a non-informative prior is derived from the square root of the determinant of the Fisher information matrix. This method has gained extensive adoption for its property of remaining invariant under bijective transformations of the parameter space. In the context of the Rice distribution, the computation involves taking the square root of the determinant of $I(\alpha,\eta)$, leading to
\begin{equation}\label{priorjnk}
\pi_2\left(\alpha,\eta\right)\propto \frac{\sqrt{(\rho+1)\Psi(\rho)-\rho}}{\alpha^2}.
\end{equation}

The joint posterior distribution for $\alpha$ and $\eta$, produced by the Jeffreys prior, is proportional to the product of the likelihood function (\ref{verori}) and the prior (\ref{priorjnk}) resulting in,
\begin{equation}\label{postjnk1} 
\begin{aligned}
p_2(\alpha,\eta|\boldsymbol{x})\propto\frac{\sqrt{(\rho+1)\Psi(\rho)-\rho}}{\alpha^{2n+2}}\prod_{i=1}^{n}\left(x_iI_0\left( \frac{\eta x_i}{\alpha^2} \right)\right) \exp\left( -\sum_{i=1}^{n}\frac{x_i^2 + \eta^2}{2\alpha^2} \right). 
\end{aligned}
\end{equation}

\begin{proposition}\label{prop2} This posterior is proper for $n> 2$ as long as not all $x_i$ are equal.
\end{proposition}
\begin{proof} From Idier \cite{idier2014properties}, $\Psi(\rho)$ is strictly increasing in $\rho$ and $\Psi(\rho)\to 1$ as $\rho\to \infty$ from which in special it follows that $\Psi(\rho)\leq 1$ for all $\rho>0$ and thus
\begin{equation*}
\sqrt{(\rho+1)\Psi(\rho)-\rho}\leq \sqrt{(\rho+1) - \rho} = 1\mbox{ for all }\rho >0.
\end{equation*}
Thus $\pi_2(\alpha,\eta)\leq \pi(\alpha,\eta)$ for $\pi(\alpha,\eta)=\frac{1}{\alpha^2}$. But applying Theorem \ref{theoprinc} with $r_0=r_\infty=0$ and $k=-2$ we conclude $\pi(\alpha,\eta)$ leads to a proper joint posterior for all
\begin{equation*}n>\max(2(r_\infty+1),k+1)=2
\end{equation*}
as not all $x_i$ are equal. And since $\pi_2(\alpha,\eta)\leq \pi(\alpha,\eta)$, the same follows for the joint posterior of $\pi_2(\alpha,\eta)$.
\end{proof}

\subsection{Sampling from the posterior}\label{secmetrol}

Here, we outline the Monte Carlo Markov chain algorithm for sampling from the joint posterior distribution.  To generate samples of $\eta$ and $\alpha$ from the marginal posterior distribution, the Metropolis-Hastings (MH) algorithm is required since the marginal distributions do not have closed-form expressions. Therefore, in order to obtain samples from the marginal distributions, we can use the conditional distribution given by 
\begin{equation}\label{cond1}
p_2(\eta|\alpha,\boldsymbol{x})\propto \sqrt{(\rho+1)\Psi(\rho)-\rho}\prod_{i=1}^{n}I_0\left( \frac{\eta x_i}{\alpha^2} \right) \exp\left( -\sum_{i=1}^{n}\frac{x_i^2 + \eta^2}{2\alpha^2} \right),\end{equation}

\begin{equation}\label{cond2}
p_2(\alpha|\eta,\boldsymbol{x})\propto\frac{\sqrt{(\rho+1)\Psi(\rho)-\rho}}{\alpha^{2n+2}}\prod_{i=1}^{n}I_0\left( \frac{\eta x_i}{\alpha^2} \right) \exp\left( -\sum_{i=1}^{n}\frac{x_i^2 + \eta^2}{2\alpha^2} \right).\end{equation}

In this study, we adopt the Gamma distribution $q(\alpha^{(*)}|\alpha^{(j)},k)$ y $q(\eta^{(*)}|\eta^{(j)},d)$ as a proposal distribution to sample values of the parameters $\alpha$ and $\eta$, respectively, where $d$ and $k$ are hyperparameters that influence the convergence rate of the algorithm. It is important to note that alternative proposal distributions can be utilized in place of the Gamma model, such as any model that generates values in the positive real line. The subsequent steps detail the execution of the MH algorithm:

\begin{enumerate}
\item Compute the initial values of $\eta^{(1)}$ and $\alpha^{(1)}$ from (\ref{mme1}) and (\ref{mme2}) and initialize a counter $j=1$;
\item Generate a random number $\eta^{(*)}$ from the $Gamma(\eta^{(j)}, d)$ distribution;
\item Calculate the acceptance probability, defined as:
\begin{equation*}
\Delta\left(\eta^{(j)},\eta^{(*)}\right)=\min\left(1, \frac{\pi\left(\eta^{(*)}|\alpha^{(j)},\boldsymbol{x}\right)}{\pi\left(\eta^{(j)}|\alpha^{(j)},\boldsymbol{x}\right)} \frac{\f{q}\left(\eta^{(j)}|\eta^{(*)},d\right)}{\f{q}\left(\eta^{(*)}|\eta^{(j)},d\right)}\right),
\end{equation*}
where $\pi(\cdot)$ denotes the conditional posterior distribution of $\eta$ given in (\ref{cond1}). Draw a random sample from an independent uniform distribution $u$ in the interval (0,1);
\item If $\Delta\left(\eta^{(j)},\eta^{(*)}\right)\geq u(0,1)$, accept the value $\eta^{(*)}$ and set $\eta^{(j+1)}=\eta^{(*)}$. If $\Delta\left(\eta^{(j)},\eta^{(*)}\right)< u(0,1)$, reject the value and set $\eta^{(j+1)}=\eta^{(j)}$;

\item Generate a random number $\alpha^{(*)}$ from the $Gamma(\alpha^{(j)}, k)$ distribution;
\item Calculate the acceptance probability, defined as:
\begin{equation*}
\Delta\left(\alpha^{(j)},\alpha^{(*)}\right)=\min\left(1, \frac{p\left(\alpha^{(*)}|\eta^{(j+1)},\boldsymbol{x}\right)}{p\left(\alpha^{(j)}|\eta^{(j+1)},\boldsymbol{x}\right)} \frac{\f{q}\left(\alpha^{(j)}|\alpha^{(*)},k\right)}{\f{q}\left(\alpha^{(*)}|\alpha^{(j)},k\right)}\right),
\end{equation*}
where $p(\cdot)$ is the conditional posterior $\alpha$ distribution given by (\ref{cond2}). Draw a random sample from an independent uniform distribution $u$ in the interval (0,1);
\item If $\Delta\left(\alpha^{(j)},\alpha^{(*)}\right)\geq u(0,1)$, accept the value $\alpha^{(*)}$ and set $\alpha^{(j+1)}=\alpha^{(*)}$. If $\Delta\left(\alpha^{(j)},\alpha^{(*)}\right)< u(0,1)$, reject the value and set $\alpha^{(j+1)}=\alpha^{(j)}$;

\item Increment the counter (j) to (j+1) and repeat steps 2-7 until the chains converge.
\end{enumerate}

\section{Simulation Study}

In this section, we employ Monte Carlo methods to conduct a simulation study to compare the impact of the non-informative priors on posterior distributions, with the aim of identifying the most efficient estimation method when compared with the moment estimator (MM) and the MLE (see the supplemental material to obtain the expressions). This is achieved by computing the Bias and mean square errors (MSE) as defined below:
\begin{equation*}
\f{Bias}{\theta_j}=\frac{1}{N}\sum_{i=1}^{N}(\hat\theta_{i,j}-\theta_j) \ \ \mbox{ and } \ \ \f{MSE}{\theta{j}}=\sum_{i=1}^{N}\frac{(\hat\theta_{i,j}-\theta_{j})^2}{N}, \quad j=1,2\end{equation*}
where $\boldsymbol{\theta}=(\alpha,\eta)$ and $N=5,000$ represent the number of samples generated for each $n$ and are used to obtain the estimates through the different methods to validate our results. The evaluation includes the 95\% coverage probability ($CP_{95\%}$) of credibility intervals (CI) and asymptotic confidence intervals for $\alpha$ and $\eta$. The optimal estimators, under this approach, exhibit both Bias and MSE closer to zero. Moreover, for a significant number of experiments at a $95\%$ confidence level, the frequencies of intervals covering the true values of $\boldsymbol{\theta}$ should approach $95\%$. We use the Bayes estimates based on the posterior mean of $\alpha$ and posterior median of $\eta$ due to the asymmetry of the posterior distributions.

The results sampling technique to obtain the samples from the posterior was generated using R software. The code is available on GitHub in the code availability section to ensure reproducibility. The details describing how to sample from the posterior and the diagnostic analysis are available jointly with the codes and are standard in Bayesian applications as are exactly the same as assumed in the simulation section in Ramos et al. \cite{ramos2018bayesian}. Considering $n=(10,15,\ldots,60)$, the outcomes are presented only for $\boldsymbol\theta=(6,2)$ due to space constraints. However, similar results are obtained for other choices of $\alpha$ and $\eta$. Figures \ref{graf22}-\ref{graf24} displays, respectively, the Bias, MSEs, and $CP_{95\%}$ from different values of $n$.
\begin{figure}[!h]
\centering
\includegraphics[scale=0.5]{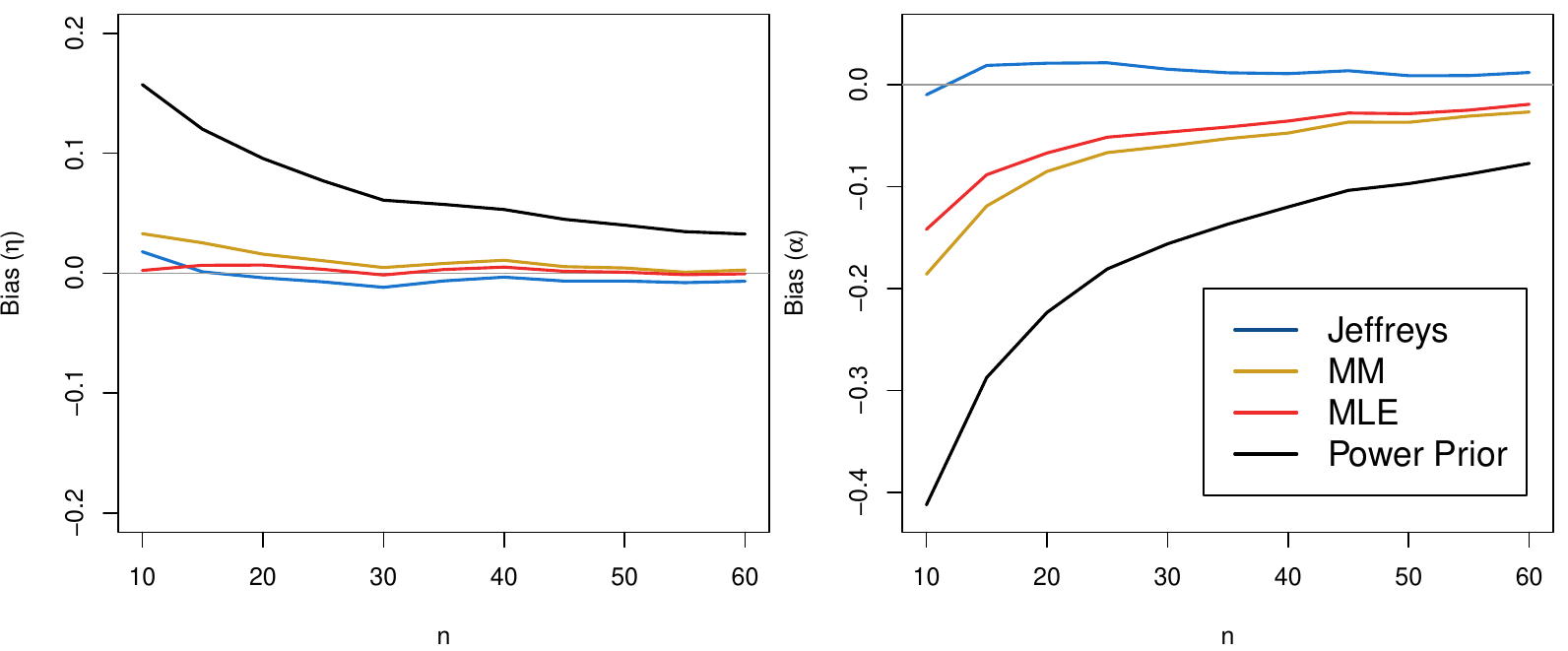}
\caption{Bias, for the estimates of $\eta =6$ and $\alpha = 2$, for $R=10,000$ simulated samples of size $n$, and using the MM, MLE, and the Bayes estimators.}
\label{graf22}
\end{figure}

As can be seen in Figure \ref{graf22}, the Bayesian approach using the prior proposed by Lauwers et al. \cite{lauwers2009estimating} returned the highest bias among the methods. On the other hand, our Bayesian approach using Jeffreys prior demonstrated significantly lower bias compared to MM and MLE. This suggests that our method can provide nearly unbiased estimates for sample sizes equal to or greater than 10, which is particularly beneficial for applications requiring high precision.

\begin{figure}[!h]
\centering
\includegraphics[scale=0.5]{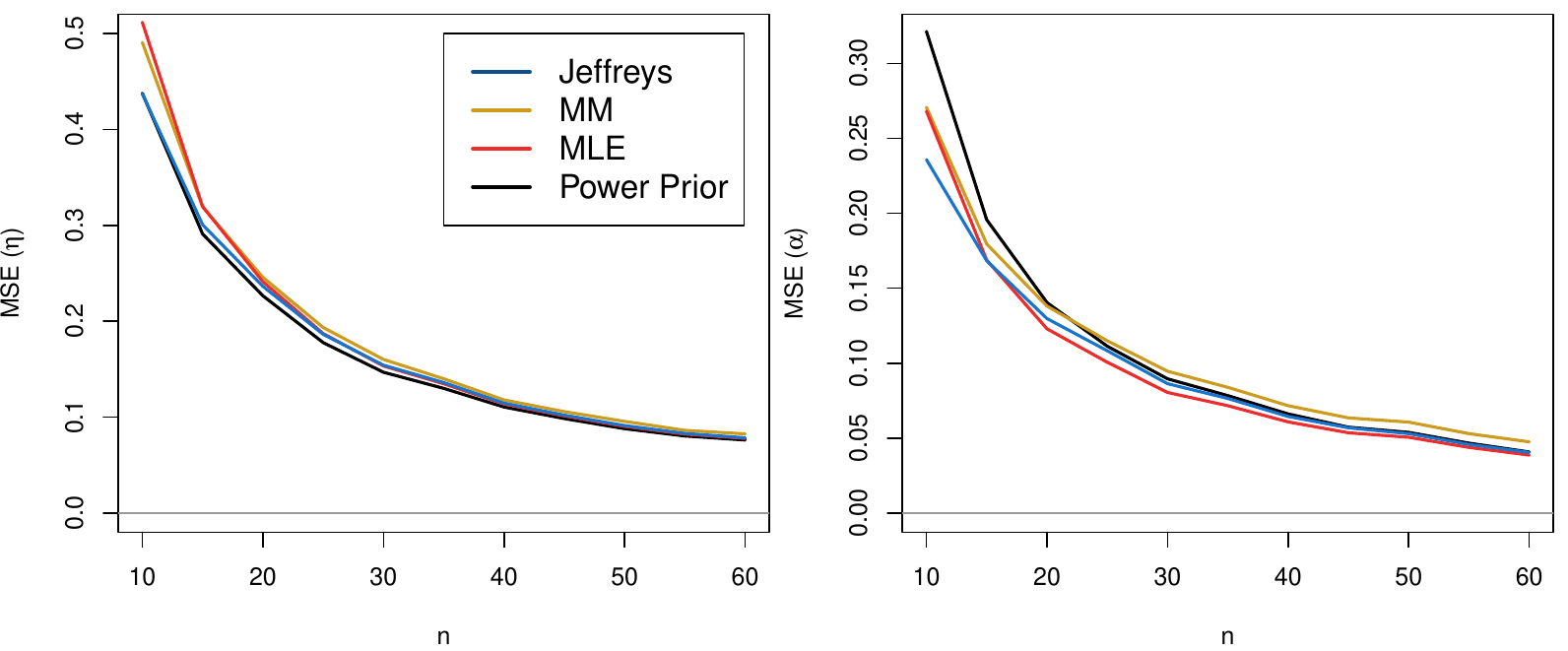}
\caption{MSE for the estimates of $\eta =6$ and $\alpha = 2$, for $R=10,000$ simulated samples of size $n$, and using the MM, MLE, and the Bayes estimators.}
\label{graf23}
\end{figure}

Figure \ref{graf23} illustrates the MSE results, where our Bayesian estimates consistently exhibit lower MSE than those obtained from other methods. This reduced error enhances the reliability of the estimations.
\begin{figure}[!h]
\centering
\includegraphics[scale=0.5]{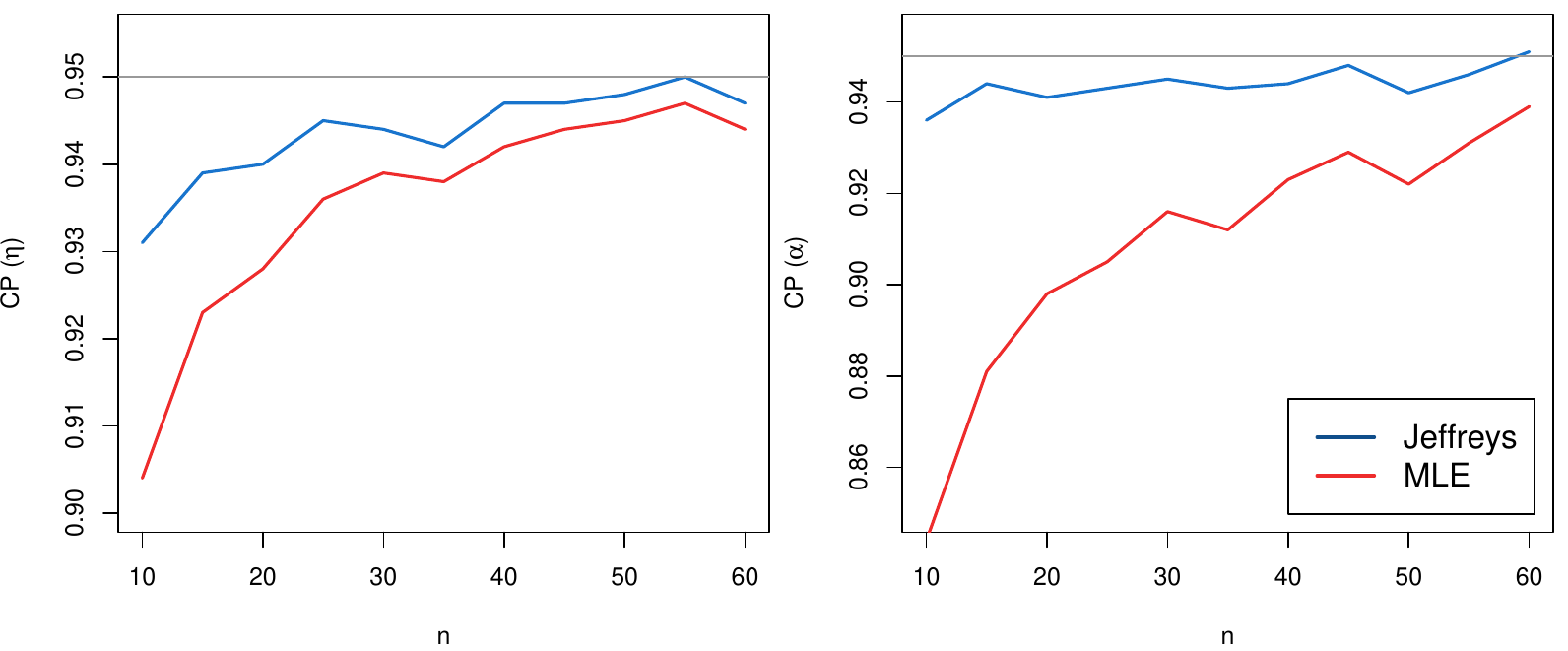}
\caption{CP for the estimates of $\eta =6$ and $\alpha = 2$, for $N=10,000$ simulated samples of size $n$, and using the MM, MLE, and the Bayes estimators.}
\label{graf24}
\end{figure}

In Figure \ref{graf24}, the coverage probabilities provided by the objective Bayesian approach are closer to the ideal 95\%, especially for smaller sample sizes. This is observed as our Bayesian approach does not require asymptotic properties (as the MLEs do) to construct the confidence/credibility intervals. The results directly come from the posterior distribution (without any normality assumption) and can be computed as long as the posterior is proper, i.e., \(n > 2\).

Overall, the simulation returned from the objective Bayesian approach with the Jeffreys prior returned accurate results for all metrics cited above and should be used to obtain accurate estimates for the parameters even for small sample sizes.

\section{Posterior Predictive Distribution}

To make predictions about new observations using our Bayesian framework, we derive the posterior predictive distribution based on the joint posterior distribution of the parameters \(\alpha\) and \(\eta\). Given the observed data \(\boldsymbol{x} = (x_1, x_2, \dots, x_n)\), the posterior predictive distribution for a new observation \(y_{\text{new}}\) is given by:
\begin{equation}\label{post_pred}
p(y_{\text{new}} \mid \boldsymbol{x}) = \int_0^\infty \int_0^\infty p(y_{\text{new}} \mid \alpha, \eta) \, p_2(\alpha, \eta \mid \boldsymbol{x}) \, d\alpha \, d\eta,
\end{equation}
where \(p(y_{\text{new}} \mid \alpha, \eta)\) is the likelihood function of the Rician distribution evaluated at \(y_{\text{new}}\), and \(p_2(\alpha, \eta \mid \boldsymbol{x})\) is the joint posterior distribution obtained using the Jeffreys prior.

Computing the integral in (\ref{post_pred}) analytically is intractable due to the complexity introduced by the Bessel function and the form of the posterior distribution. Therefore, we utilize the samples \(\{ (\alpha^{(j)}, \eta^{(j)}) \}_{j=1}^{N}\) generated from the joint posterior distribution via the Metropolis-Hastings algorithm outlined in Section \ref{secmetrol}. For each sample \((\alpha^{(j)}, \eta^{(j)})\), we generate a corresponding predictive sample \(y_{\text{new}}^{(j)}\) by drawing from the Rician distribution:
\begin{equation}\label{predictive_sample}
y_{\text{new}}^{(j)} \sim \text{Rice}(\alpha^{(j)}, \eta^{(j)}).
\end{equation}

The collection \(\{ y_{\text{new}}^{(j)} \}_{j=1}^{N}\) constitutes an empirical approximation of the posterior predictive distribution. From this set, we compute predictive value and its credible intervals:
\begin{equation}\label{predictive_mean}
\hat{y}_{\text{mean}} = \frac{1}{N} \sum_{j=1}^{N} y_{\text{new}}^{(j)},
\end{equation}
and the \((1 - \gamma)\times100\%\) credible interval is obtained from the \(\gamma/2\) and \(1 - \gamma/2\) quantiles of the \(y_{\text{new}}^{(j)}\) samples.

By utilizing the predictive samples, we account for the uncertainty in the parameter estimates \(\alpha\) and \(\eta\) when making predictions about new data points. This Bayesian predictive approach provides a probabilistic framework that naturally incorporates parameter uncertainty into the predictions.

In practice, after obtaining the MCMC samples \(\{ (\alpha^{(j)}, \eta^{(j)}) \}\), the predictive sampling proceeds as follows:

\begin{enumerate}
    \setcounter{enumi}{7}
    
    \item Generate a predictive sample \( y_{\text{new}}^{(j)} \) from the Rician distribution with parameters \(\alpha^{(j)}\) and \(\eta^{(j)}\).
    \item  Increment the counter (j) to (j+1) and repeat steps 2-8 until the chains converge.
\end{enumerate}

Predictive modeling using Bayesian frameworks is advancing applications in areas like wireless communications, finance \cite{giudici2023safe}, and healthcare. However, complex models can suffer from issues like lack of explainability, robustness, and potential biases, leading to unreliable predictions and unfair outcomes. To address these challenges, the S.A.F.E. principles—Sustainable, Accurate, Fair, and Explainable—have been established \cite{giudici2024safe}. These principles ensure that predictive models are not only precise but also robust to uncertainties, equitable across different data groups, and transparent in their methodology.

Implementing the S.A.F.E. principles in predictive analysis involves developing consistent statistical metrics to assess model compliance \cite{babaei2025rank}. For example, in our Bayesian estimation of the Rician distribution for wireless communication systems, using non-informative priors like the Jeffreys prior results in proper posterior distributions and reliable parameter estimates. By generating posterior predictive distributions through methods like the Metropolis-Hastings algorithm, researchers can account for parameter uncertainty, enhancing the model's accuracy and explainability. Such Bayesian approaches align with the S.A.F.E. framework by ensuring that predictive models are robust, fair, and transparent, thereby improving their practical applicability and trustworthiness in critical domains.

\newpage

\section{Application}

Outage probability in wireless communications refers to the probability that the signal-to-noise ratio (SNR) drops below a critical threshold, resulting in a failure to maintain reliable communication. This concept is fundamental in assessing the performance of wireless systems, especially in environments where signal fading plays a significant role.The outage probability \(P_{\text{out}}\) for a Rician-faded signal is given by the following integral:

\[
P_{\text{out}} = \int_0^{\gamma_{\text{th}}} f_{\text{Rician}}(\gamma; \eta, \alpha) \, d\gamma
\]
where \( \gamma_{\text{th}} \) is the SNR threshold below which the communication is considered to be in outage. The function \( f_{\text{Rician}}(\gamma; \eta, \alpha) \) represents the probability density function of the Rician distribution, already presented in the introduction. Table \ref{tablem2} presents a sample of signal noise data received from a wireless system, modeled using a Rice distribution. The true parameter values for this distribution are set to $\eta = 5$ and $\alpha = 2$ while \( n = 35 \), reflecting typical characteristics of signal fading in wireless communications.

\begin{table}[!h]
\centering
\caption{A sample of signal noise received from a wireless system with a Rice distribution, where the true values are $\eta = 5$ and $\alpha = 2$. }
\begin{tabular}{c|c|c|c|c|c|c}
\hline
2.3860 & 2.7988 & 2.8369 & 2.9939 & 3.7689 & 3.7791 & 4.8047 \\ \hline
4.8163 & 4.9192 & 5.0718 & 5.3107 & 5.4689 & 5.5971 & 5.9371 \\ \hline
6.4130 & 6.6319 & 6.6981 & 6.7382 & 6.8516 & 7.0193 & 7.0738 \\ \hline
7.0869 & 7.3947 & 7.4064 & 7.5912 & 7.6672 & 7.8145 & 7.8603 \\ \hline
8.0524 & 8.1097 & 8.2362 & 8.2652 & 8.8715 & 9.0293 & 9.5223 \\ \hline
\end{tabular}\label{table3}

\end{table}

\begin{figure}[!h]
\centering
\includegraphics[scale=0.8]{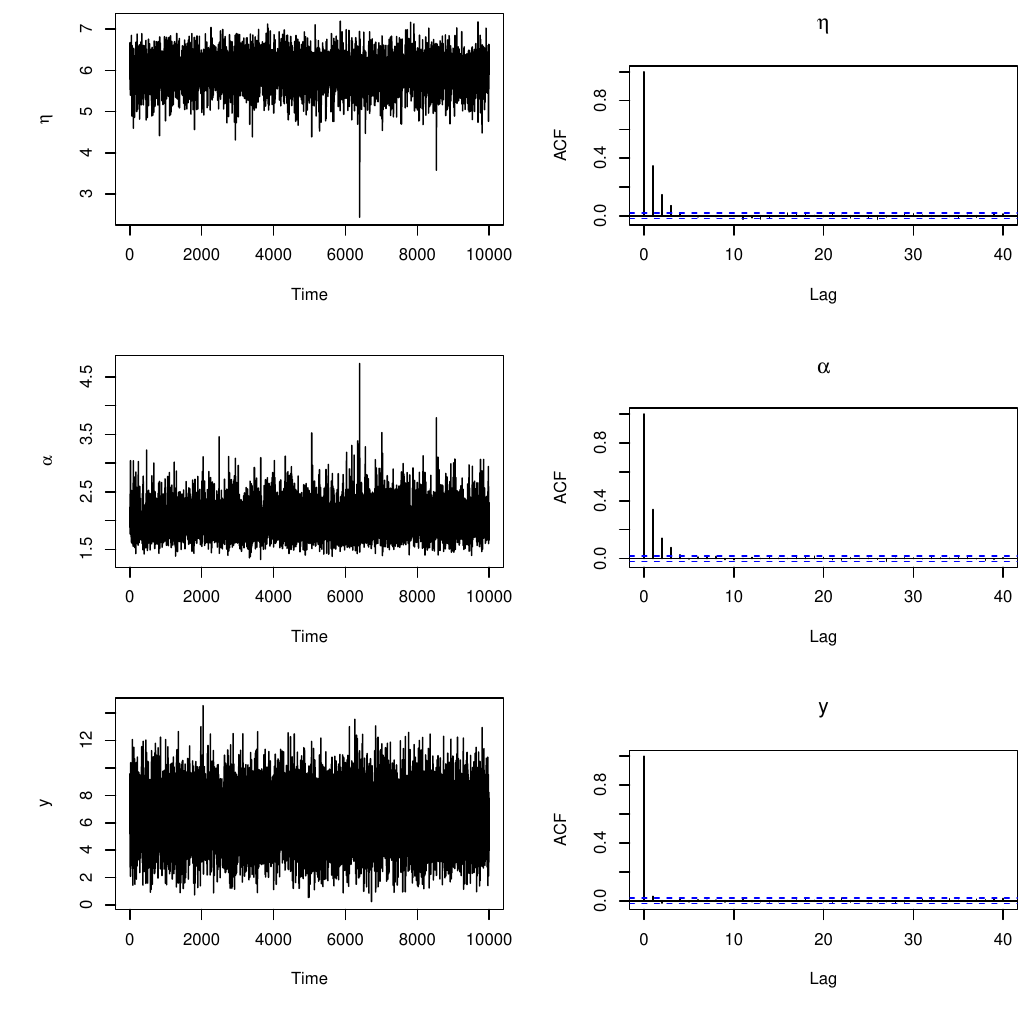}
\caption{ Left panel, the trace for the estimated parameters via MCMC algorithm. Right panel, the autocorrelation function plot for the parameter.}
\label{graf3}
\end{figure}

 Using the function available in the Code Availability section, we generated a chain of 50,500 samples, discarding the initial 500 as burn-in. We applied thinning with a factor of 5, resulting in two chains of 10,000 samples each. Figure \ref{graf3} displays the convergence plots for the estimated parameters, including the trace plots obtained through the MCMC algorithm and the autocorrelation function plots for selected parameters.

Additionally, we applied the Geweke diagnostic to verify the convergence of all the chains. Table \ref{tablem2} provides the 95\% credible intervals (equal-tailed) for each estimate, along with the Bayesian estimates of the parameters.
\begin{table}[!ht]
\caption{Bayes estimates, standard deviations and $95\%$ credible intervals for $\eta,\alpha$ and $\alpha$ from the data.}
\centering 
\begin{center}
  \begin{tabular}{ c | c | c | c}
    \hline
		$\boldsymbol{\theta}$  & Bayes & SD & CI$_{95\%}(\boldsymbol{\theta})$ \\ \hline
    \ \ $\eta$ \ \   & 5.961 & 0.379  & (5.167; 6.673) \\ \hline
    \ \ $\alpha$   \ \  & 2.004 & 0.275  & (1.562; 2.626) \\ \hline
     \ \ $y_{new}$   \ \  & 6.333 & 1.954  & (2.522; 10.255) \\ \hline
  \end{tabular}\label{tablem2}
\end{center}
\end{table}

In wireless communications, this outage probability describes the likelihood that the signal-to-noise ratio (SNR) will drop below a critical threshold, leading to communication failure. The Rician distribution is particularly effective in environments where both a line-of-sight (LOS) path and multiple scattered signals are present. characteristics.

By computing the integral up to \( \gamma_{\text{th}} \), we estimate the outage probability for different SNR thresholds. Comparing the true outage with our fitted estimates and their corresponding credibility intervals demonstrates the precision and robustness of our approach, validating its effectiveness for predicting the performance of wireless systems under varying fading conditions.
As an application in wireless communication, in Figure \ref{graf3} we present the estimate of the outage probability using the Rice distribution for $\eta = 6,$ $\alpha = 2$, and $n = 35$, under different SNR thresholds. Here, we compare the true outage with the fit obtained through our approach and its respective credibility interval, demonstrating high precision in the estimation.
\begin{figure}[!h]
\centering
\includegraphics[scale=0.5]{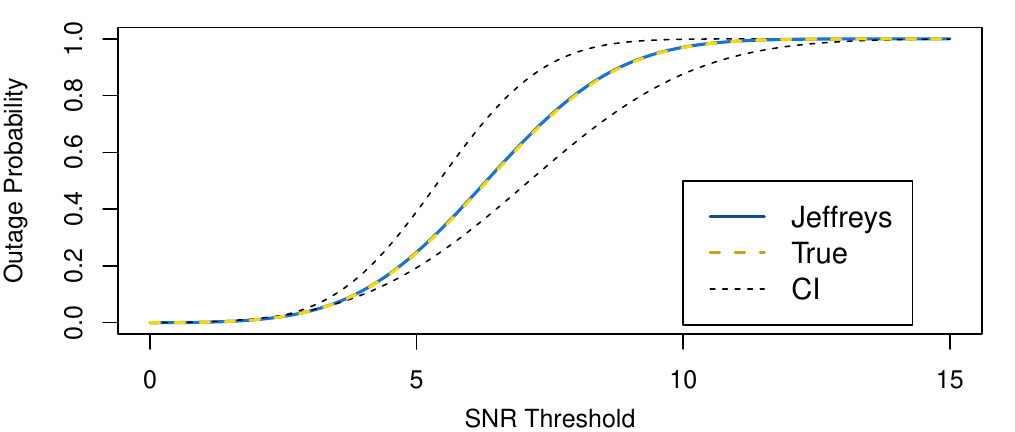}
\caption{Outage Probability vs. SNR Threshold for Rice Distribution with $\eta = 6,$ $\alpha = 2$, and $n = 35$, using Bayes estimators.}
\label{graf3}
\end{figure}

As can be seen in  Figure \ref{graf3}, conducting proper inference is crucial for accurately modeling and predicting wireless communication performance, especially when considering critical metrics like outage probability. As demonstrated by the results, the Bayesian estimates, along with the credible intervals, show a high degree of precision, validating the robustness of our approach.

\section{Conclusions}

In this paper, we investigated Bayesian estimation for the parameters of the Rician distribution using non-informative priors, specifically focusing on the Jeffreys prior. Our analysis demonstrates that the Jeffreys prior results in a proper posterior distribution, thus providing a more robust and reliable method for parameter estimation in this model. This contrasts with the power prior, which was found to be biased and less suitable for the Rician distribution. The efficacy of our approach is supported by a comprehensive simulation study, which shows that Bayesian estimates obtained with the Jeffreys prior exhibit lower bias and mean squared error (MSE) compared to classical methods such as maximum likelihood estimators (MLEs) and moment-based estimators. The Bayesian estimators under the Jeffreys prior produced nearly unbiased estimates, as indicated by bias and MSE values approaching zero.

Overall, the application of our methodology to the estimation of outage probability in wireless systems demonstrates its practical relevance. The precise Bayesian estimates obtained in this study not only highlight the efficacy of the proposed approach but also emphasize the importance of robust statistical inference in engineering applications, where reliable parameter estimation is crucial for system performance evaluation. In summary, this work contributes to the growing body of research on Bayesian methods for parameter estimation, particularly for the Rician distribution, and offers a foundation for future investigations into more complex statistical models and applications.

Our results are particularly significant given the challenges associated with deriving expressions for Fisher information in this model and, consequently, its Jeffreys prior. As the Fisher information depends on complex integrals and can not be factorized, as discussed in Bernardo \cite{bernardo2005}, reference priors were difficult to obtain. However, further investigation should be considered with the aim of deriving them. The obtained priors can be easily applied to our main theorem to confirm if it leads to proper posteriors. Additionally, future research could explore the extension of our methodology to more complex scenarios frequently encountered in practice. This includes cases with censoring, the incorporation of covariates, and analyses involving multivariate data. Investigating the adaptability of Bayesian approaches in these more intricate settings remains a promising direction for future work.

\section*{Code Availability}

The Metropolis-Hastings algorithm for sampling from the posterior distribution has been implemented in a package that is easy to use and available at: \\ \url{https://jeachire.github.io/riccib/reference/riccibo.html}


\bibliographystyle{tfs}

\begin{appendix}    

\section{Proof of Theorem \ref{theoprinc}}

Since the integrand below is always positive from the Fubbini-Tonelli Theorem (See Folland, \cite{folland}) we have
\begin{equation*}
\begin{aligned}
\int_{\mathcal{A}} p_1(\alpha,\eta|\boldsymbol{x})\, d\boldsymbol{\theta} \propto \int_{\mathcal{A}} 
 \frac{\pi(\alpha)\pi(\eta)}{ \alpha^{2n}}\prod_{i=1}^{n}\left(x_iI_0\left( \frac{\eta x_i}{\alpha^2} \right)\right) \exp\left( -\sum_{i=1}^{n}\frac{x_i^2 + \eta^2}{2\alpha^2} \right)\, d\boldsymbol{\theta}\\
= \int_0^\infty \int_0^\infty\frac{\pi(\alpha)\pi(\eta)}{\alpha^{2n}}\prod_{i=1}^{n}\left(x_iI_0\left( \frac{\eta x_i}{\alpha^2} \right)\right) \exp\left( -\sum_{i=1}^{n}\frac{x_i^2 + \eta^2}{2\alpha^2} \right)\, d\eta\, d\alpha\\
\propto \int_0^\infty \int_0^\infty\frac{\pi(\alpha)\pi(\eta)}{\alpha^{2n}}\prod_{i=1}^{n}I_0\left( \frac{\eta x_i}{\alpha^2} \right) \exp\left( -\sum_{i=1}^{n}\frac{x_i^2 + \eta^2}{2\alpha^2} \right)\, d\eta\, d\alpha.
\end{aligned}
\end{equation*}

\noindent\textbf{Proof of item $i)$} Since the first term of the series of $I_0(\eta)$ in (\ref{dens}) is $1$ it follows that $I_0(\eta) \geq 1$ for $\eta>0$. Thus
\begin{equation*}
\begin{aligned}
\int_0^\infty \int_0^\infty\frac{\pi(\alpha)\pi(\eta)}{\alpha^{2n}}\prod_{i=1}^{n}I_0\left( \frac{\eta x_i}{\alpha^2} \right) \exp\left( -\sum_{i=1}^{n}\frac{x_i^2 + \eta^2}{2\alpha^2} \right)\, d\eta\, d\alpha\\
\geq \prod_{i=1}^{n}\int_0^\infty 
 \int_0^\infty\frac{\pi(\alpha)\pi(\eta)}{ \alpha^{2n}}\exp\left( -\sum_{i=1}^{n}\frac{x_i^2 + \eta^2}{2\alpha^2} \right)\, d\eta \, d\alpha.
\end{aligned}
\end{equation*}

Now, for fixed $\alpha>0$ we have $\exp\left( -\sum_{i=1}^{n}\frac{x_i^2 + \eta^2}{2\alpha^2} \right)\underset{\eta \to 0^+}{\propto} 1$. Therefore, from Proposition \ref{proportional2}, since $\pi(\eta)\underset{\eta\to 0^+}{\propto}\eta^{r_0}$ with $r_0\leq -1$ we have for fixed $\alpha>0$ that
\begin{equation*} \int_0^1 \frac{\pi(\alpha)\pi(\eta)}{\alpha^{2n}} \exp\left( -\sum_{i=1}^{n}\frac{x_i^2 + \eta^2}{2\alpha^2} \right)\, d\eta \propto  \int_0^1 \eta^{r_0}\, d\eta = \infty.
\end{equation*}
Thus $\int_0^\infty\frac{\pi(\alpha)\pi(\eta)}{ \alpha^{2n}}\exp\left( -\sum_{i=1}^{n}\frac{x_i^2 + \eta^2}{2\alpha^2} \right)\, d\eta=\infty$ for each $\alpha>0$ and thus
\begin{equation*}
\begin{aligned} \int_{\mathcal{A}} p_1(\alpha,\eta|\boldsymbol{x})\, d\boldsymbol{\theta} \gtrsim \int_0^\infty 
 \int_0^\infty\frac{\pi(\alpha)\pi(\eta)}{ \alpha^{2n}}\exp\left( -\sum_{i=1}^{n}\frac{x_i^2 + \eta^2}{2\alpha^2} \right)\, d\eta \, d\alpha = \int_0^\infty \infty \, d\alpha = \infty.
\end{aligned}
\end{equation*}

\noindent\textbf{Proof of item $ii)$} Denoting $I_0^*\left(y\right) = e^{-y} I_0\left(y\right)$ for all $y>0$, and since $\pi(\alpha)\propto \alpha^k$ we have
\begin{equation*}
\begin{aligned}
\int_0^\infty \int_0^\infty \frac{\pi(\eta)\pi(\alpha)}{\alpha^{2n}}\prod_{i=1}^{n}I_0\left( \frac{\eta x_i}{\alpha^2} \right) \exp\left( -\sum_{i=1}^{n}\frac{x_i^2 + \eta^2}{2\alpha^2} \right)\, d\alpha\, d\eta\\
=\int_0^\infty \int_0^\infty \frac{\pi(\eta)}{\alpha^{2n-k}}\prod_{i=1}^{n}I_0^*\left( \frac{\eta x_i}{\alpha^2} \right) \exp\left( -\sum_{i=1}^{n}\frac{(x_i - \eta)^2}{2\alpha^2} \right)\, d\alpha\, d\eta.
\end{aligned}
\end{equation*}
Now, considering the change of variables  $\alpha = \sqrt{\frac{\eta}{\beta}} \Leftrightarrow d\alpha = - \frac{1}{2}\sqrt{\frac{\eta}{\beta^3}}d\beta$ we obtain:
\begin{equation*}
\begin{aligned}
&\int_0^\infty \int_0^\infty \frac{\pi(\eta)}{\alpha^{2n-k}}\prod_{i=1}^{n}I_0^*\left( \frac{\eta x_i}{\alpha^2} \right) \exp\left( -\sum_{i=1}^{n}\frac{(x_i - \eta)^2}{2\alpha^2} \right)\, d\alpha\, d\eta\\
&= \frac{1}{2}\int_0^\infty \int_0^\infty \frac{\pi(\eta)\beta^{n-\frac{(k+3)}{2}}}{\eta^{n-\frac{(k+1)}{2}}}\prod_{i=1}^{n}I_0^*\left(\beta x_i \right) \exp\left( -\beta\sum_{i=1}^{n}\frac{(x_i - \eta)^2}{2\eta} \right)\, d\beta\, d\eta\\
&= s_1 + s_2 + s_3,
\end{aligned}
\end{equation*}
where
\begin{equation*} h(\beta,\eta,\boldsymbol{x}) = \frac{\pi(\eta)\beta^{n-\frac{(k+3)}{2}}}{\eta^{n-\frac{(k+1)}{2}}}\prod_{i=1}^{n}I_0^*\left(\beta x_i \right) \exp\left( -\beta\sum_{i=1}^{n}\frac{(x_i - \eta)^2}{2\eta} \right),
\end{equation*}
and
\begin{equation*}
\begin{aligned}
s_1 =& \frac{1}{2}\int_0^\infty  
 \int_0^1 h(\beta,\eta,\boldsymbol{x}) \, d\beta\, d\eta,
\\ s_2 =& \frac{1}{2}\int_1^\infty  
 \int_1^\infty h(\beta,\eta,\boldsymbol{x}) \, d\beta\, d\eta ,\ s_3 = \frac{1}{2}\int_0^1  
 \int_1^\infty h(\beta,\eta,\boldsymbol{x}) \, d\beta\, d\eta.
\end{aligned}
\end{equation*}

Now, from Abramowitz 
and Stegun (\cite{abramowitz}, pp. 375-377) we have
\begin{equation*}
\begin{aligned}
I_0(y) \underset{y\to 0^+}{\propto} 1,\ I_0(y) \underset{y\to \infty}{\propto} \frac{e^y}{\sqrt{y}},
\end{aligned}
\end{equation*}
and thus, letting $I_0^*\left(y\right) = e^{-y} I_0\left(y\right)$, since $e^{-y} \underset{y\to 0^+}{\propto} 1$ it follows that
\begin{equation}\label{I_0_ineq}
\begin{aligned}
I_0^*(y) \underset{y\to 0^+}{\propto} 1,\ I_0^*(y) \underset{y\to \infty}{\propto} \frac{1}{\sqrt{y}}.
\end{aligned}
\end{equation}
Thus, in special $\prod_{i=1}^{n}I_0^*\left(\beta x_i \right)\underset{\beta\to\infty}{\propto} \frac{1}{\sqrt{\beta^{{n}}}}$ and we have
\begin{equation*}
\begin{aligned}
h(\beta,\eta,\boldsymbol{x}) \underset{\beta\to 0^+} \propto \frac{\pi(\eta)\beta^{n-\frac{(k+3)}{2}}}{\eta^{n-\frac{(k+1)}{2}}} \exp\left( -\beta\sum_{i=1}^{n}\frac{(x_i - \eta)^2}{2\eta} \right)\mbox{ and}
\\ h(\beta,\eta,\boldsymbol{x}) \underset{\beta\to \infty} \propto \frac{\pi(\eta)\beta^{\frac{n-k-3}{2}}}{\eta^{n-\frac{(k+1)}{2}}} \exp\left( -\beta\sum_{i=1}^{n}\frac{(x_i - \eta)^2}{2\eta} \right).
\end{aligned}
\end{equation*}
Therefore, from Proposition \ref{proportional2}, we have
\begin{equation*}
\begin{aligned}
s_1 \propto \int_0^\infty  
 \int_0^1 \frac{\pi(\eta)\beta^{n-\frac{(k+3)}{2}}}{\eta^{n-\frac{(k+1)}{2}}} \exp\left( -\beta\sum_{i=1}^{n}\frac{(x_i - \eta)^2}{2\eta} \right) \, d\beta\, d\eta\\
 \leq \int_0^\infty \int_0^\infty \frac{\pi(\eta)\beta^{n-\frac{(k+3)}{2}}}{\eta^{n-\frac{(k+1)}{2}}} \exp\left( -\beta\sum_{i=1}^{n}\frac{(x_i - \eta)^2}{2\eta} \right) \, d\beta\, d\eta,
\end{aligned}
\end{equation*}
and from the change of variables $\gamma = \beta \sum_{i=1}^{n}\frac{(x_i - \eta)^2}{2\eta}\Leftrightarrow d\beta = \frac{2\eta}{\sum_{i=1}^{n}(x_i - \eta)^2}d\gamma$ and therefore, since $n - \frac{k+1}{2}=\frac{n}{2} + \frac{n-k-1}{2}> 0$ we have
\begin{equation*}
\begin{aligned}
s_1& \lesssim \int_0^\infty  
 \int_0^\infty \frac{\pi(\eta)\eta^{n-\frac{(k+1)}{2}}\gamma^{n-\frac{(k+3)}{2}}}{(\sum_{i=1}^{n}(x_i - \eta)^2)^{n-\frac{k+1}{2}}\eta^{n-\frac{(k+1)}{2}}} \exp\left( -\gamma\right) \, d\gamma\, d\eta
 \\ &
 = \int_0^\infty 
  \frac{\pi(\eta)}{(\sum_{i=1}^{n}(x_i - \eta)^2)^{n-\frac{k+1}{2}}} \int_0^\infty \gamma^{n-\frac{(k+1)}{2}-1}\exp\left( -\gamma\right) \, d\gamma\, d\eta
  \\ & = \Gamma\left(n-\frac{k+1}{2}\right)\int_0^\infty 
  \frac{\pi(\eta)}{(\sum_{i=1}^{n}(x_i - \eta)^2)^{n-\frac{k+1}{2}}}\, d\eta\\
 & = \Gamma\left(n-\frac{k+1}{2}\right)\left(\int_0^1 
  \frac{\pi(\eta)}{(\sum_{i=1}^{n}(x_i - \eta)^2)^{n-\frac{k+1}{2}}}\, d\eta + \int_1^\infty 
  \frac{\pi(\eta)}{(\sum_{i=1}^{n}(x_i - \eta)^2)^{n-\frac{k+1}{2}}}\, d\eta \right).
\end{aligned}
\end{equation*}
But since not all $x_i$ are equal we have $\sum_{i=1}^n (x_i-\eta)^2>0$ for all $\eta>0$ and moreover
\begin{equation*}\frac{1}{\sum_{i=1}^n (x_i-\eta)^2}\underset{\eta\to 0^+} {\propto }1\mbox{ and }
\frac{1}{\sum_{i=1}^n (x_i-\eta)^2}\underset{\eta\to \infty} {\propto }\frac{1}{\eta^2}.
\end{equation*}
Thus, from Proposition \ref{proportional2}, since $r_0>-1$ and $r_\infty - 2n < \frac{n}{2}-1 - 2n <0$ we have
\begin{equation*} s_1 \lesssim \Gamma\left(n+\frac{1}{2}\right)\left(\int_0^1 \eta^{r_0} d\eta + \int_1^\infty \eta^{r_\infty-2n-1} d\eta\right) <\infty,
\end{equation*}
and thus $s_1$ is finite.
Now, on the other hand
\begin{equation*}
\begin{aligned}
s_2 \propto& \int_1^\infty 
 \int_1^\infty h(\beta,\eta,\boldsymbol{x}) \, d\beta\, d\eta \\
  \propto& \int_1^\infty  
\int_1^\infty \frac{\pi(\eta)\beta^{\frac{n-k-3}{2}}}{\eta^{n-\frac{(k+1)}{2}}} \exp\left( -\beta\sum_{i=1}^{n}\frac{(x_i - \eta)^2}{2\eta} \right)\, d\beta\, d\eta\\
\leq& \int_1^\infty  
\int_0^\infty \frac{\pi(\eta)\beta^{\frac{n-k-3}{2}}}{\eta^{n-\frac{(k+1)}{2}}} \exp\left( -\beta\sum_{i=1}^{n}\frac{(x_i - \eta)^2}{2\eta} \right)\, d\beta\, d\eta,
 \end{aligned}
\end{equation*}
and, from the change of variables $\gamma = \beta \sum_{i=1}^{n}\frac{(x_i - \eta)^2}{2\eta}$ from which $d\beta = \frac{2\eta}{\sum_{i=1}^{n}(x_i - \eta)^2}d\gamma$, and since by hypothesis $n-k-1>0$
\begin{equation*}
\begin{aligned}
s_2 &\lesssim \int_1^\infty  
 \int_0^\infty \frac{\pi(\eta)\eta^{\frac{n-k-1}{2}}\gamma^{\frac{n-k-3}{2}}}{(\sum_{i=1}^{n}(x_i - \eta)^2)^{\frac{n-k-1}{2}}\eta^{n-\frac{k+1}{2}}} \exp\left( -\gamma\right) \, d\gamma\, d\eta
 \\
 &= \int_1^\infty 
  \frac{\pi(\eta)\eta^{-\frac{n}{2}}}{(\sum_{i=1}^{n}(x_i - \eta)^2)^{\frac{n-k-1}{2}}} \int_0^\infty \gamma^{\frac{n-k-1}{2}-1}\exp\left( -\gamma\right) \, d\gamma\, d\eta
  \\ &= \Gamma\left(\frac{n-k-1}{2}\right)\int_1^\infty 
  \frac{\pi(\eta)\eta^{-\frac{n}{2}}}{(\sum_{i=1}^{n}(x_i - \eta)^2)^{\frac{n-k-1}{2}}}\, d\eta,
\end{aligned}
\end{equation*}
and, as before $\frac{1}{\sum_{i=1}^n (x_i-\eta)^2}\underset{\eta\to \infty} {\propto }\frac{1}{\eta^2}$, since $r_\infty+k-\frac{3n}{2}+2<\frac{3n}{2}-2 - \frac{3n}{2}+2 = 0$
\begin{equation*} s_2 \lesssim \Gamma\left(\frac{n+1}{2}\right) \int_1^\infty \frac{\eta^{r_\infty-\frac{n}{2}}}{\eta^{n-k-1}} d\eta = \int_1^\infty \eta^{(r_\infty+k-\frac{3n}{2}+2)-1} d\eta <\infty,
\end{equation*}
and thus $s_2$ is finite. Finally, we have from the hypothesis and the second proportionality in (\ref{I_0_ineq}) that
\begin{equation*}
\begin{aligned}
s_3 &\propto \int_0^1  
 \int_1^\infty h(\beta,\eta,\boldsymbol{x}) \, d\beta\, d\eta \\
  &\propto \int_0^1  
\int_1^\infty \frac{\beta^{n-\frac{(k+3)}{2}}}{\eta^{n-\frac{(k+1)}{2}-r_0}} \exp\left( -\beta\sum_{i=1}^{n}\frac{(x_i - \eta)^2}{2\eta} \right)\, d\beta\, d\eta,
 \end{aligned}
\end{equation*}
and, from the change of variables $\gamma = \beta \frac{\sum_{i=1}^{n}(x_i - \eta)^2}{2\eta}$ from which $d\beta = \frac{2\eta}{\sum_{i=1}^{n}(x_i - \eta)^2}d\gamma$, we have
\begin{equation*}
\begin{aligned}
s_3 &\lesssim \int_0^1 
  \frac{1}{\eta^{\frac{n}{2}-r_0}(\sum_{i=1}^{n}(x_i - \eta)^2)^{\frac{n-k-1}{2}}}\int_{\frac{\sum_{i=1}^{n}(x_i - \eta)^2}{2\eta}}^\infty \gamma^{\frac{n-k-1}{2}-1}\exp\left( -\gamma\right) \, d\gamma\, d\eta
  \\ &= \int_0^1 \frac{\Gamma\left(\frac{n-k-1}{2},\frac{\sum_{i=1}^{n}(x_i - \eta)^2}{2\eta}\right)}{\eta^{\frac{n}{2}-r_0}(\sum_{i=1}^{n}(x_i - \eta)^2)^{\frac{n-k-1}{2}}}\, d\eta =  \int_0^1 T(\eta)\, d\eta,
\end{aligned}
\end{equation*}
where here $\Gamma(s,x)=\int_x^\infty t^{s-1}e^{-t}\, dt$ is the upper incomplete gamma function and
\begin{equation*}
\begin{aligned}
T(\eta) = \frac{\Gamma\left(\frac{n-k-1}{2},\frac{\sum_{i=1}^{n}(x_i - \eta)^2}{2\eta}\right)}{\eta^{\frac{n}{2}-r_0}(\sum_{i=1}^{n}(x_i - \eta)^2)^{\frac{n-k-1}{2}}}\mbox{ for all }\eta>0\mbox{ and }T(0) = 0.
\end{aligned}
\end{equation*}

However, from Abramowitz 
and Stegun (\cite{abramowitz} p. 263) we have for all fixed $s>0$ that $\Gamma(s,x)\underset{x\to \infty}{\propto} x^{s-1}e^{-x}$ and thus since $h_2(\eta)=\frac{\sum_{i=1}^{n}(x_i - \eta)^2}{2\eta}$ is positive with 
$\lim_{\eta\to 0^+} h_2(\eta)=\infty$ it follows that
\begin{equation*} \Gamma\left(\frac{n-k-1}{2},\frac{\sum_{i=1}^{n}(x_i - \eta)^2}{2\eta}\right) \underset{\eta\to 0^+}{\propto} h_2(\eta)^{\frac{n-k-3}{2}}e^{-h_2(\eta)},
\end{equation*}
and since $\lim_{\eta\to 0^+}\frac{\eta}{h_2(\eta)^{-1}} = \frac{\sum_{i=1}^n x_i^2}{2}>0$ it follows that $\eta \underset{\eta\to 0^+}{\propto} h_2(\eta)^{-1}$ and since $(\sum_{i=1}^{n}(x_i - \eta)^2)^{\frac{n+1}{2}} \underset{\eta\to 0^+}{\propto} 1 $ we conclude that
\begin{equation*}  T(\eta) = \frac{\Gamma\left(\frac{n-k-1}{2},\sum_{i=1}^{n}\frac{(x_i - \eta)^2}{2\eta}\right)}{\eta^{\frac{n}{2}-r_0}(\sum_{i=1}^{n}(x_i - \eta)^2)^{\frac{n-k-1}{2}}}  \underset{\eta\to 0^+}{\propto} \frac{h_2(\eta)^{\frac{n-k-3}{2}}e^{-h_2(\eta)}}{h_2(\eta)^{-\frac{n}{2}+r_0}} = h_2(\eta)^{n-\frac{k+3}{2}-r_0}e^{-h_2(\eta)},
\end{equation*}
additionally, as $\lim_{\eta\to 0^+}h_2(\eta) = \infty$ it follows using change of variables $u=h_2(\eta)$ under the limit that
\begin{equation*}\lim_{\eta\to 0^+} h_2(\eta)^{n-\frac{k+3}{2}-r_0}e^{-h_2(\eta)} = \lim_{u\to \infty} u^{n-\frac{k+3}{2}-r_0}e^{-u} = 0,
\end{equation*}
and thus it follows from the proportionality above that 
\begin{equation*}
\lim_{\eta\to 0^+} T(\eta) = 0 = T(0),
\end{equation*}
that is, $T(\eta)$ is continuous in $0$. Thus, the function $T(\eta)$ is continuous in $[0,1]$ making it integrable on $[0,1]$. Consequently, we can conclude that
\begin{equation*}
\begin{aligned}
s_3 \lesssim \int_0^1 
  T(\eta)\, d\eta < \infty
\end{aligned}
\end{equation*}
which concludes the proof.

\end{appendix}

\end{document}